\nonstopmode
\documentclass[letterpaper,11pt]{article}
\usepackage[english]{babel}
\usepackage[utf8]{inputenc}
\usepackage{lmodern}
\usepackage{lipsum}
\usepackage{bm}
\usepackage{wrapfig}
\usepackage{graphicx}
\usepackage{caption}
\usepackage{hyperref}
\usepackage[colon]{natbib}
\usepackage{tocvsec2}
\usepackage{subcaption}
\usepackage{float}
\usepackage{soul}
\usepackage{longtable}
\usepackage{threeparttablex}
\usepackage{url}

\usepackage{ragged2e}
\justifying

\usepackage{dsfont}

\usepackage[nice]{nicefrac}
\usepackage{amsmath}
\usepackage{xfrac}

\usepackage[bottom,hang,flushmargin]{footmisc}

\usepackage{rotating, graphicx}

\usepackage{parskip}

\usepackage{titling}
\usepackage{dcolumn}
\newcolumntype{d}[1]{D{.}{.}{#1}}

\usepackage{graphicx}
\usepackage{multirow}
\usepackage{tabularx}
\usepackage{tabulary}
\usepackage{tabu}
\usepackage{amsmath}
\usepackage[absolute,overlay]{textpos}
\usepackage{wasysym} 

\usepackage[toc,header]{appendix}
\usepackage{inputenc}
\usepackage{pifont}
\usepackage{fancyvrb}
\usepackage{multicol}
\usepackage[normalem]{ulem}
\usepackage{comment}
\usepackage{array}
\usepackage{booktabs}
\usepackage[labelfont=bf]{caption}
\usepackage{pdflscape}
\usepackage{placeins}
\usepackage{setspace}
\usepackage{epstopdf}
\usepackage{lipsum}
\usepackage{setspace}
\usepackage{csquotes}
\usepackage{bbm}
\usepackage{amsmath}
\usepackage{amssymb}
\usepackage{amsthm}
\usepackage{mathtools}
\usepackage[T1]{fontenc}
\usepackage{setspace}
\usepackage{pdfpages}
\usepackage{graphicx}
\usepackage{adjustbox}
\usepackage[margin=1in]{geometry}

\definecolor{darkred}{RGB}{141, 0, 0}
\usepackage{xcolor}
\hypersetup{
    colorlinks,
    linkcolor={darkred},
    citecolor={darkred},
    urlcolor={darkred}
}

\newtheorem{theorem}{Theorem}

\begin{document}
\title{Resolving the automation paradox: falling labor share, rising wages\thanks{We thank Keelan Beirne, Maarten Goos, Ramayya Krishnan, Alan Manning, Ezra Oberfield, Lowell Taylor and Anna Salomons for comments that improved the paper, and Raimundo Contreras for research assistance. Autor acknowledges research support from the Google Technology and Society Visiting Fellows Program, the NOMIS Foundation, the Schmidt Sciences AI2050 Fellowship, the Smith Richardson Foundation, and the James M. and Kathleen D. Stone Foundation.}}

\thispagestyle{empty}
\author{{\parbox{4.5cm}{\centering David Autor\\MIT and NBER}\hspace{0.cm}\parbox{4.5cm}{\vspace{0.1cm}\centering B.N. Kausik\\Independent}}\bigskip}
 
\clearpage\maketitle
\thispagestyle{empty}

\begin{abstract}
    \noindent A central socioeconomic concern about Artificial Intelligence is that it will lower wages by depressing the labor share---the fraction of economic output paid to labor. We show that declining labor share is more likely to raise wages. In a competitive economy with constant returns to scale, we prove that the wage-maximizing labor share depends only on the capital-to-labor ratio, implying a non-monotonic relationship between labor share and wages. When labor share exceeds this wage-maximizing level, further automation increases wages even while reducing labor's output share. Using data from the United States and eleven other industrialized countries, we estimate that labor share is too high in all twelve, implying that further automation should raise wages. Moreover, we find that falling labor share accounted for 16\% of U.S. real wage growth between 1954 and 2019. These wage gains notwithstanding, automation-driven shifts in labor share are likely to pose significant social and political challenges.



\end{abstract}

\pagenumbering{arabic}


\section*{Main}

Labor share, the fraction of national income (or GDP) that is paid to workers as wages, salaries, or benefits, is declining in industrialized countries \cite{karabarbounis2024perspectives}. According to Bureau of Labor Statistics (BLS) data, the labor share in the US 
fell from 67.8\% to 58.4\% between 1987 and 2019, Figure \ref{fig:1}. Leading explanations for this intensively scrutinized phenomenon include (1) technological change \cite{karabarbounis2014, acemoglu2022tasks, OberfieldRaval, jones2024framework, karabarbounis2024perspectives}; (2) globalization \cite{Elsbyetal2013}; (3) changes in market structure \cite{Autoretal2020, barkai2020}; (4) declining union bargaining power \cite{bentolila2003, benmelech2022}; and (5) demographic changes. \cite{glover2020, Hopenhayn2018, kausik2023cognitive}

\begin{figure}
    \centering
    \includegraphics[width=0.5\linewidth]{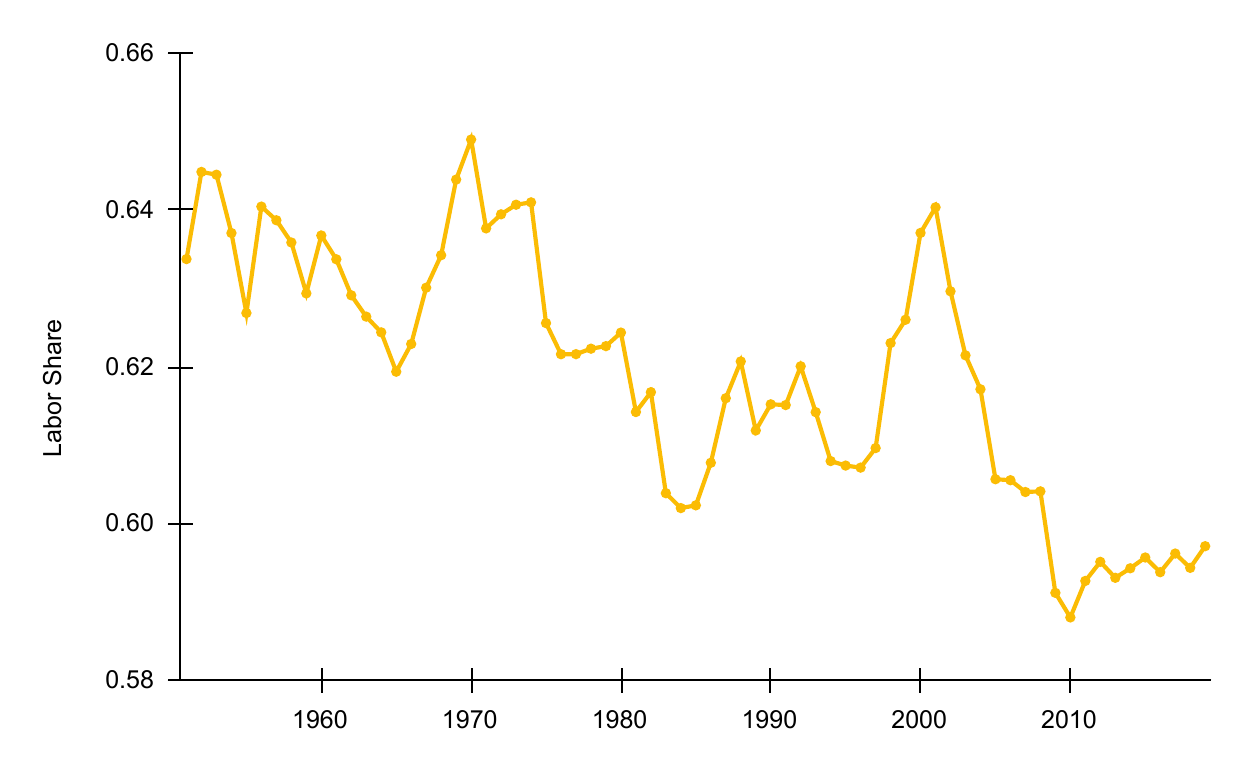}
    \caption{\textbf{Share of labor compensation in gross domestic product \\ at current national prices in the United States, 1950--2019}}
    \label{fig:1}
    \medskip
    \raggedright
    \small 
    {\footnotesize \textit{Note.} Figure is based on Federal Reserve Bank of St. Louis series LABSHPUSA156NRUG and Feenstra, Robert C., Robert Inklaar and Marcel P. Timmer (2015), ``The Next Generation of the Penn World Table'' \textit{American Economic Review}, 105(10), 3150-3182.}
\end{figure}

Against this backdrop, there is rising concern that rapid advances in Artificial Intelligence may accelerate the downward trend in labor share across numerous countries, displacing human labor and depressing wages.\cite{korinek2018artificial, acemoglu2025simple, brynjolfsson2025canaries} The challenge is particularly acute if automation proves ``so-so,''\cite{Acemoglurestrepo2019} meaning that it displaces workers across a broad range of tasks with minimal improvements in productivity, benefiting capital but not labor.

We show that these multiple concerns regarding the trajectory of wages, automation, productivity, and labor share boil down to a simple metric, the capital-labor ratio. Specifically, we prove that in a competitive economy operating at constant returns to scale, the maximum wage and the labor share that maximizes wages depend only on the capital-to-labor ratio. Paradoxically, when the prevailing labor share is above the wage-maximizing level, increasing automation that reduces labor share will increase wages. The relevance of this theoretical result is confirmed through analysis of economic data from twelve industrialized countries. We estimate that a further increase in automation with a corresponding decline in labor share would cause average wages to rise rather than fall. 

The social and political implications of our results are not straightforward. A fall in the labor share that increases average wages may nevertheless raise inequality and reduce wages in some sectors. Moreover, a falling labor share necessarily means increased income and wealth inequality due to the highly uneven distribution of capital ownership. Concentration of wealth and income might in turn potentially magnify political economic risks such as the resource curse\cite{ploeg2011natural}.
     
\section*{Resolving the automation-wage paradox }

Let $Y(K,L,\lambda)$ be a family of smooth production functions for economic output $Y$ from inputs capital $K$ and labor $L$ respectively, at labor share $\lambda$. Our novel formulation of the production function is general, encompassing any set of technologies that combine labor and capital to produce output. Smoothness guarantees that the function is differentiable. Constant returns to scale means that the wage can be written as $w=\lambda (Y/L)$, equal to total output scaled by the labor share and divided by the quantity of labor employed. 

\medskip

\begin{theorem}\label{gentheorem}
Consider a family of constant-returns-to-scale production functions $Y(K,L,\lambda)$. Then: \textbf{(a)} The wage-maximizing labor share and the maximum wage itself depend only on the capital-labor ratio $k = K/L$. \textbf{(b)} Suppose there exists at least one labor share $\lambda \in (0,1)$ satisfying the first-order condition: $\lambda = -1/ \left(\partial \ln Y / \partial \lambda \right)$. If at each such $\lambda$, the second-order condition holds $\frac{\partial^2 \ln Y}{\partial \lambda^2} < 0$, then the wage-maximizing labor share is unique and lies strictly between 0 and 1.
\end{theorem}

\begin{proof}
\textbf{Part (a):} By constant returns to scale, 
$$Y/L = Y((K/L),(L/L),\lambda) = Y(k,1,\lambda),$$
where $k \equiv K/L$ is the capital-to-labor ratio. Therefore, wages are given by
\begin{equation}\label{eq:gen_wage}
    w = \lambda(Y/L) = \lambda Y(k,1,\lambda)
\end{equation}
and 
\begin{equation}\label{eq:log_gen_wage}
   \ln w = \ln \lambda +\ln Y(k,1,\lambda).
\end{equation}
Differentiating with respect to $\lambda$,
\begin{equation}\label{eq:log_deriv}
   \frac{\partial \ln w} {\partial \lambda} = \frac{1}{\lambda} +  
   \frac{\partial \ln Y(k,1,\lambda )}{\partial \lambda}.
\end{equation}
Setting $\partial \ln w/\partial\lambda = 0$, we obtain the first-order condition
\begin{equation}\label{eq:lambda_s}
   \lambda = -1/ \left(\partial \ln Y(k,1,\lambda )/\partial \lambda \right).
\end{equation}
The solution(s) to Equation \eqref{eq:lambda_s} depend only on the capital-to-labor ratio $k$. Some of these solutions may be maxima, while others may be minima. Note that $w(0) = 0$ and, at the other boundary, the wage at $\lambda = 1$ is $w(1) = Y(k,1,1)$, which also depends only on $k$. Therefore, wages are maximized either at an interior solution to Equation \eqref{eq:lambda_s} or at one of the boundaries, and the wage-maximizing labor share $\lambda^{*}$ depends only on the capital-to-labor ratio. Substituting $\lambda^{*}$ into Equation \eqref{eq:gen_wage}, it follows that the maximum wage also depends only on $k$:
$$w^{*} = \lambda^{*}Y(k,1,\lambda^{*}).$$

\textbf{Part (b):} By assumption, there exists at least one interior solution $\tilde{\lambda} \in (0,1)$ of Equation \eqref{eq:lambda_s}, i.e., $\tilde{\lambda} = -1/(\partial \ln Y/\partial \lambda).$ We now verify that each such $\tilde{\lambda}$ is a global maximum. 

Differentiating Equation \eqref{eq:log_deriv} with respect to $\lambda$,
\begin{equation}\label{eq:second_deriv}
  \frac{\partial^2 \ln w}{\partial \lambda^2} = -\frac{1}{\lambda^2} + \frac{\partial^2 \ln Y}{\partial \lambda^2}.
\end{equation}
At any critical point $\tilde{\lambda}$ satisfying Equation \eqref{eq:lambda_s}, we have $\frac{1}{\tilde{\lambda}^2} > 0$ and, by assumption, $\frac{\partial^2 \ln Y}{\partial \lambda^2}\big|_{\tilde{\lambda}} < 0$. Therefore,
$$\frac{\partial^2 \ln w}{\partial \lambda^2}\bigg|_{\tilde{\lambda}} = -\frac{1}{\tilde{\lambda}^2} + \frac{\partial^2 \ln Y}{\partial \lambda^2}\bigg|_{\tilde{\lambda}} < 0,$$
confirming that each interior critical point is a strict local maximum.

To establish uniqueness, suppose for contradiction that there exist two distinct interior local maxima at $\lambda_1, \lambda_2 \in (0,1)$ with $\lambda_1 < \lambda_2$. By continuity of $w(\lambda)$ and the fact that $w$ attains local maxima at both points, there must exist at least one local minimum at some $\lambda_m \in (\lambda_1, \lambda_2)$. At this local minimum, the first-order condition \eqref{eq:lambda_s} must hold: $\frac{\partial \ln w}{\partial \lambda}\big|_{\lambda_m} = 0$. However, we have shown that every solution to Equation \eqref{eq:lambda_s} satisfies the second-order condition for a strict local maximum. This contradicts the fact that $\lambda_m$ is a local minimum. Therefore, there can be at most one interior critical point.

Since we assume the existence of at least one interior solution to \eqref{eq:lambda_s}, and we have shown there can be at most one, the wage-maximizing labor share is a unique global maximum with $\lambda^{*} \in (0,1)$.
\end{proof}

Equation \eqref{eq:log_deriv} reflects two competing effects of an increase in the labor share on wages: the first term on the right represents the direct effect of increasing labor share on wages, which is necessarily positive; and the second term on the right represents the indirect effect of increasing labor share on wages through its effect on output. This term is negative at the critical point but may be positive or negative elsewhere. Per Equation \eqref{eq:second_deriv}, a critical point is a maximum if the indirect wage effect (falling output) dominates the direct effect (higher labor share of output) as labor share rises beyond the critical point. 

\section*{Validation on historical data}

We quantify our theoretical results by applying the workhorse Constant Elasticity of Substitution production function\cite{Arrowetal1962, McFadden1963} to economic data from the United States and eleven other industrialized countries. This production function takes the form:
\begin{equation}\label{eq:ces-prod-fn}
Y = A\left(\alpha  K^{\rho} + (1-\alpha)  L^{\rho}\right)^{1 / \rho}, 
\end{equation}
where $\alpha$ is a parameter in the unit interval $\left[0,1\right]$, $\sigma = 1/\left(1-\rho\right)$ is the elasticity of substitution between labor and capital, and $A$ is Total Factor Productivity (TFP), which is independent of other parameters of the model. The parameter $\alpha$ in Equation \eqref{eq:ces-prod-fn} is often referred to as the `share' parameter, referring not to the capital \textit{share of output} as in our discussion above but the \textit{capital share of tasks} in the production function. 
Equation \eqref{eq:ces-prod-fn} can be rewritten as below relating the output per labor hour $y \equiv Y/L$ to the capital-to-labor ratio $k$, elasticity of substitution $\sigma$ and the labor share $\lambda$
\begin{equation}
      y = A\left(\lambda + (1-\lambda)k^{(1-\sigma)/\sigma}\right)^{(\sigma/(1-\sigma))},
\end{equation}

where $\lambda = \left(1-\alpha\right)/\left(\alpha k^\rho + (1-\alpha)\right)$. Numerous studies \cite{hamermesh1993labor, katz-murphy1992,  chirinko2008sigma,  knoblach2020elasticity} report values of $\sigma$ in the range of $0.5$ to $1.5$.  Taking logarithms and expanding $\lambda + (1-\lambda)k^{(1-\sigma)/\sigma}$ as a Taylor series in $x =(1-\sigma)/\sigma$ around $x=0$ (a valid approximation when $\sigma$ is near $1$) yields: 
\begin{align}\label{ces_taylor}
    \ln (y) &= \ln (A) +\frac{1}{x}\left((1-\lambda)\ln(k)x + \frac{1}{2!}(1-\lambda)\lambda (\ln(k))^2x^2+ ... \right )   \notag \\
    &\approx \ln (A) + (1-\lambda)\ln(k) + \frac{1}{2}(1-\lambda)\lambda (\ln(k))^2x.
\end{align}

This expression for $y$, equal to output per labor hour, provides a basis for analyzing historical data from twelve countries: the United States, Austria, Belgium, Denmark, France, Germany, Italy, Japan, the Netherlands, Spain, Sweden and the United Kingdom. Data are sourced from the 2023 EU KLEMS database \cite{euklems2023} supplemented with U.S.-specific data sources, as reported in Table \ref{tab:data_sources}. For each country $j$, we construct time series for output $Y_{jt}$, total labor hours $L_{jt}$, wages $w_{jt}$, labor share $\lambda_{jt}$, Total Factor Productivity (TFP) $A_{jt}$ normalized to one in the base year, and the capital-to-labor ratio $k_{jt}$. To create comparable capital-to-labor units, we convert non-dollar, non-Euro currencies to Euros at the average weekly closing rates for 2015, which is the reference year for the chained volumes of the data series (1 Danish Krone = 0.13404€; 1 Yen = 0.00744€; 1 Swedish Krona = 0.107€; and 1 GBP = 1.3761€). 

We fit Ordinary Least Squares models for output per labor hour to Equation \eqref{ces_taylor} as below,
\begin{equation}\label{eq:country-fit}
    \ln \left(\frac{Y_{jt}}{L_{jt}A_{jt}}\right) - (1-\lambda_{jt})\ln(k_{jt})= \alpha_{0j}   
     + \frac{1}{2}\alpha_{1j} \lambda_{jt}(1-\lambda_{jt}) (\ln k_{jt})^2 +\epsilon_{jt},
\end{equation}

where $\hat{\alpha}_{1j}$ is an estimate of $(1-\sigma_j)/\sigma_j$ for country $j$, and standard errors are clustered at the country level. We have adequate data to support such an estimate for the US. For the other eleven countries, we pool country-by-year data to obtain a common $\hat{\alpha}_0$ and $\hat{\alpha}_{1}$. Point estimates for $\alpha_0$ and $\alpha_{1}$, reported in Table \ref{tab:fit_alphas}, provide a normalized time series that relates real output per labor hour, adjusted for TFP, to the evolution of the capital stock.  

We can rearrange our estimates of Equation \eqref{eq:country-fit} and differentiate with respect to $\lambda$ to obtain estimates of the effect of labor share on wages:
  \begin{equation}\label{deriv_fit}
   \left.\frac{ \partial \ln (w )}{\partial \ln (\lambda)} \right |_{jt} =  
   1 -  \lambda_{jt} \ln(k_{jt})\left(1   
     -(\hat{\alpha}_{1j}/2)(1-2\lambda_{jt} ) \ln (k_{jt}) \right ) 
\end{equation}

Figure \ref{fig:wage_deriv} reports the estimated wage effect of a fall in the labor share ($-\partial \ln(w) \partial \ln(\lambda)$ for each country, operating through the direct effect, which is necessarily negative, the indirect effect, which will be positive if $\lambda < \lambda^*$, and the sum of these two effects. The standard error in Equation \eqref{deriv_fit} can be estimated via the standard error in $\alpha_{1j}$, with 95\% confidence intervals shown as error bars in the figure. We note that since the derivative of wages with respect to labor share fluctuates within countries over time as a function of $\lambda_{jt}$ and $k_{jt}$, we report the mean of each series for each country. We obtain 95\% confidence intervals for Equation \eqref{deriv_fit} for each country $j$ by calculating the derivative $\left.\frac{ \partial \ln (w )}{\partial \ln (\lambda)} \right |_{j}$ at $\hat{\alpha}_{1j}\pm1.96 \, \text{SE}(\hat{\alpha}_{1j})$.

\begin{figure}
    \centering
    \includegraphics[width=0.75\linewidth]{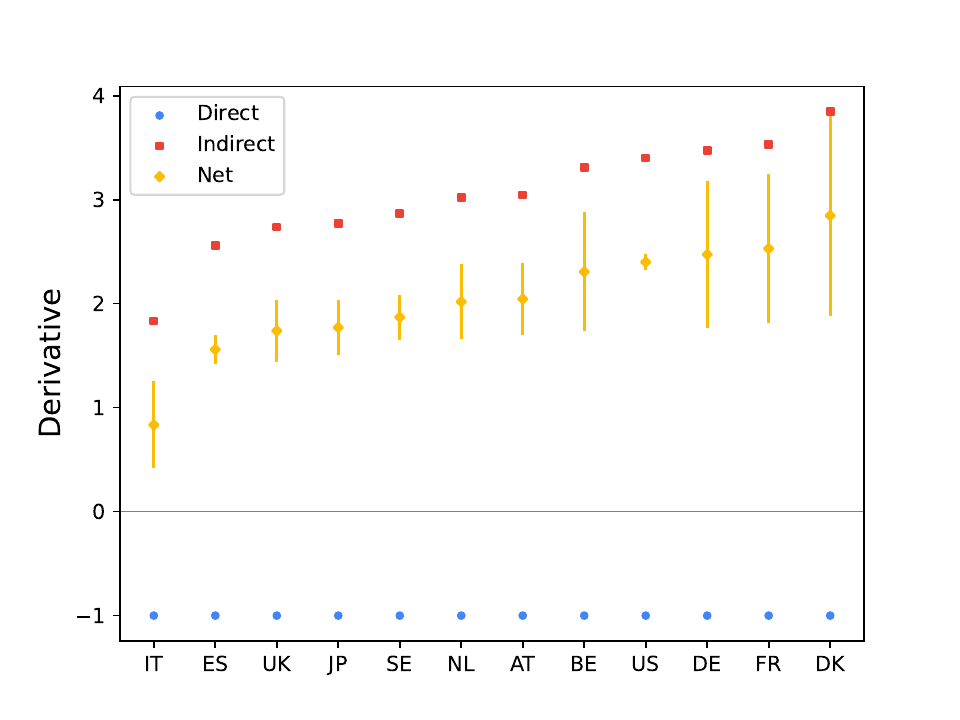}
  \captionsetup{justification=justified, width=0.9\linewidth}
    \caption{\textbf{Estimated Average $-\frac{\partial ln(w)}{\partial \ln(\lambda)}$ for Select Economies.}}
     \medskip
    \begin{justify}
    \small    
    Figure reports the estimated components of Equation \eqref{deriv_fit}. Dots correspond to the estimated direct wage effect of reducing labor labor share; squares correspond to the estimated indirect effect on wages via increased output; and diamonds correspond to the net of these two effects. Countries are ordered by the estimated net increase in wages from a one point decline in labor share. Error bars are $\pm 1.96$ times the standard error of Equation \eqref{deriv_fit}.
    \end{justify}
    \label{fig:wage_deriv}
\end{figure}

\begin{figure}[ht]
    \centering
       
    \includegraphics[width=0.75\linewidth]{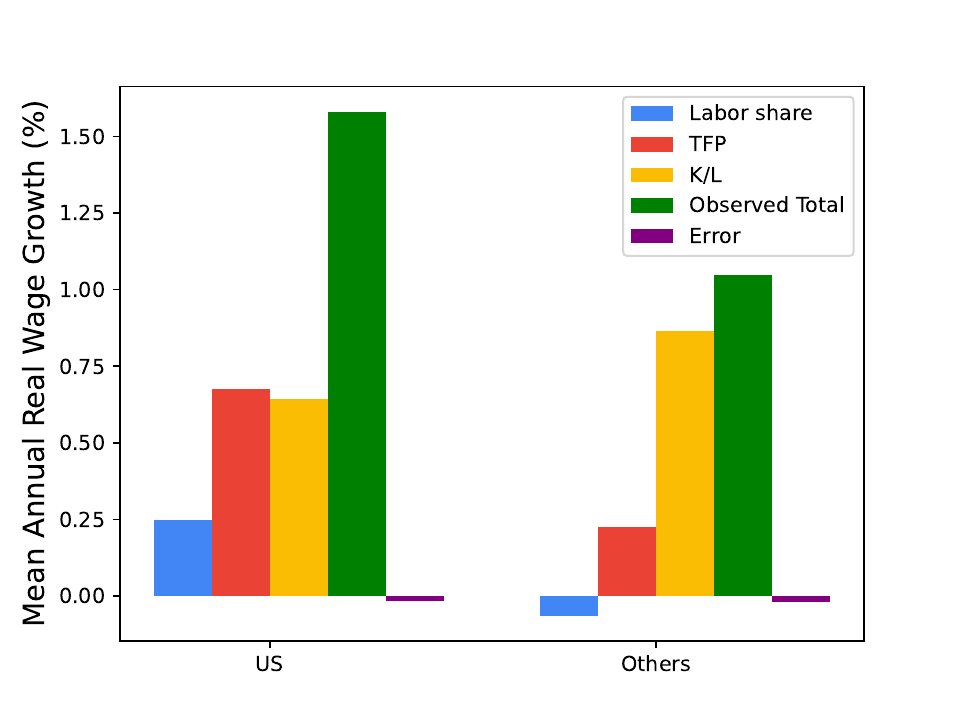}
    \captionsetup{justification=justified, width=0.9\linewidth}
    \caption{\textbf{Observed and Estimated Annual Real Wage Growth by Component for Select Economies.}}
     \medskip
    \begin{justify}
    \small    
    Figure reports the components of Equation \eqref{wage_growth} averaged over the respective time spans per Table \ref{tab:data_sources}, using the point estimates of Table \ref{tab:fit_alphas}. Also shown is the actual observed wage growth. For the United States, the sum of the estimated growth components is within $\approx 1\%$ of the observed actual and the declining labor share of Figure \ref{fig:1} contributed an estimated $\approx 15.7\%$ of the observed growth between 1954 and 2019.  For the other countries as a pool, the sum of the estimated growth components is within $\approx 2\%$ of the observed actual, and the change in labor share contributed $-6.1\%$  of the observed growth between 1995 and 2021. In the United States, labor share declined $0.1\%$ annually between 1954 to 2019, while in the other countries, the average labor share  increased $0.06\%$ annually between 1995 and 2021.
    \end{justify}
      
    \label{fig:wage_growth_by_source}
\end{figure}

Our estimates of ${\alpha}_0$ and ${\alpha}_{1}$ suggest that at current values of the capital-to-labor ratio $K/L$, a fall in the labor share would raise the average wage in all twelve countries, with the smallest effect in Italy and the largest in Denmark. The estimate of $-\partial \ln w/ \partial \ln \lambda = 2.40 \, (\text{SE}=0.040)$ for the U.S. is plausibly the most reliable since it is based on 66 years of data, versus 26 years for other countries. 

Using Equation \eqref{eq:country-fit}, we can decompose observed wage growth into three components as below:
\begin{align}\label{wage_growth}
    \frac{d\ln(w)}{dt} = & \frac{d\ln (A)}{dt} + \notag \\
     &\frac{d\ln (\lambda)}{dt} \left[1  - \lambda\ln(k)(1 -\hat{\alpha}_{1j}/2)(1-2\lambda) \ln (k)\right]  + \notag \\
    & \frac{d\ln(k)}{dt}\left [(1-\lambda)(1 + \hat{\alpha}_{1j} \lambda\ln(k) \right ] 
\end{align}
The first term on the right reflects the contribution to wage growth from total factor productivity growth; the second term reflects the contribution of the changing labor share; and the third reflects the effect of capital deepening. Figure \ref{fig:wage_growth_by_source} reports the components of Equation \eqref{wage_growth} as the mean of forward and backward differences averaged over the respective time spans per Table \ref{tab:data_sources}, using the point estimates of Table \ref{tab:fit_alphas}. Also shown is observed wage growth and a residual term, equal to the sum of the derived components from Equation \eqref{wage_growth} minus the observed change. For the U.S., the sum of the estimated growth components is within $\approx 1\%$ of the observed actual. For the other countries as a pool, the sum of the estimated growth components is also within $\approx 2\%$ of the observed actual. 

While TFP growth and capital deepening explain the bulk of annual wage growth in the U.S., we estimate that the \textit{falling} labor share (Figure \ref{fig:1}) contributed $\approx15.7\%$ of the observed growth between 1954 and 2019. In contrast to the United States, where the labor share fell on average by $0.1\%$ annually between 1954 to 2019, the labor share rose in the other countries considered between 1995 and 2021 (at an average annual rate of $0.06\%$). The rising labor share is thus estimated to have cumulatively \textit{retarded} wage growth by $\approx 6.1\%$ between 1995 and 2021 relative to the observed increase in this set of countries.

\section*{Discussion and policy implications}

Our main result is that for a competitive economy operating at constant returns to scale, the wage-maximizing labor share depends only on the capital-to-labor ratio. One implication of this result is that, while a declining labor share is frequently viewed as an ominous indicator of falling labor demand, technologies or policies that increase labor share can decrease wages and vice versa. 

Our estimates across twelve industrialized economies suggest that their prevailing labor share is above the wage-maximizing level in all of them, implying that a technological change that reduced labor share would boost wages. In the best case scenario, higher wages would offer workers more disposable income to invest in capital markets so as to participate in the returns to capital's larger share of the output. However, if labor does not own a significant portion of the capital stock, many political economy perils may arise, with potential repercussions for the stability of democratic governments.\cite{ploeg2011natural} Mineral-rich economies are often controlled by monarchs or a few state-sponsored oligarchs, and typically have low labor shares, wide disparities in wealth, and poor records of democratic freedom \cite{guerriero2019democracy}. A counterexample is Norway, which is a mineral-rich democracy that invests profits from oil revenues into a sovereign wealth fund that plausibly benefits the majority of its citizens.  

We finally note that our analysis assumes a competitive economy in which changes in labor share are caused by process or product innovations, or by capital accumulation (assuming $\sigma > 1$). If observed declines in labor share are instead driven by rent extraction, \cite{barkai2020, ace-rest-2024rents} exercise of market power, \cite{Autoretal2020, de2020rise} or failure of competition policy, \cite{gutierrez2017declining, philippon2019great} we would expect a falling labor share to lower wages directly without a countervailing positive effect arising from a higher marginal product of capital.

\begin{table}
    \centering
    \caption{\textbf{Data Sources}}
    \begin{tabular}{ll} \\
         \multicolumn{1}{c}{\textbf{Variable}}& \multicolumn{1}{c}{\textbf{Source}} \\ \hline
         \addlinespace
         \multicolumn{2}{c}{\textbf{United States (1954:2019)}}\\ 
         \addlinespace
         Real GDP& FRED GDPC1\\ 
         Total factor productivity & FRED RTFPNAUSA632NRUG\\  
         Real capital stock K& FRED RKNANPUSA666NRUG\\  
         Total labor hours& FRED B4701C0A222NBEA\\  
         Labor Share& FRED LABSHPUSA156NRUG\\ \addlinespace

        \multicolumn{2}{c}{\textbf{Japan (1995:2020)}}\\  
        \addlinespace
        Real GDP& KLEMS 2023\\ 
        Total factor productivity & FRED RTFPNAJPA632NRUG\\  
        Real capital stock K& KLEMS 2023\\ 
        Total labor hours& KLEMS 2023\\ 
        Labor Share& KLEMS 2023\\ 
        \addlinespace
        \multicolumn{2}{c}{\textbf{UK and EU countries (1995:2021)}}\\
        \addlinespace
        Real GDP& KLEMS 2023\\ 
        Total factor productivity & KLEMS 2023\\ 
        Real capital stock K& KLEMS 2023\\ 
        Total labor hours& KLEMS 2023\\ 
        Labor Share& KLEMS 2023
     \\ \hline\end{tabular}
    \label{tab:data_sources}
\end{table}

\bigskip

\begin{table}
    \addtolength{\tabcolsep}{-2 pt}
    \large
    \centering
    \caption{\textbf{Country-Specific Fit Parameters for Select Economies.}}
    \smallskip
    \begin{tabular}{ccccc} 
        &  & US&& Others \vspace{2bp}\\ \hline 
        $\alpha_0$&& 1.366&& 0.205\\ 
        && (0.065)&& (0.988)\\ \addlinespace  
        $\alpha_1$&& 0.087&&0.392\\
        && (0.020)&& (0.299)\\ \addlinespace 
        $\hat{\sigma}$&& 0.920&& 0.719\\
 & & (0.017)& &(0.196)\\ \addlinespace
        $R^2$&& 0.228&& 0.153\\ \addlinespace
        N && 65 && 323 \\ \hline
    \end{tabular}

    \begin{justify}
    \smallskip
    \small    
    Table reports point estimates for Equation \eqref{eq:country-fit} for 12 countries: Austria (AT), Belgium (BE), Germany (DE), Denmark (DK), Spain (ES), France (FR), Italy (IT), Japan (JP), the Netherlands (NL), Sweden (SE), the United Kingdom (UK), and the United States (US). Data coverage: U.S. 1954-2019; Japan 1995-2020; all other countries 1995-2021. Standard errors, reported in parenthesis, are clustered by country for the pool.
    \end{justify}
    \label{tab:fit_alphas}
\end{table}

\clearpage

\bibliographystyle{aer}
\bibliography{labor-share-wage.bib}

\end{document}